\documentclass[table]{ccjnl}
\usepackage{lipsum,amsmath}
\usepackage{cuted}
\usepackage{bm}
\usepackage{pifont,cases}
\usepackage{amsmath}
\graphicspath{{figures/}}

\title{An Envy-Free Online UAV Charging Scheme with Vehicle-Mounted Mobile Wireless Chargers}
\author{Yuntao~Wang\inst{1}, Zhou~Su\inst{1,*}\corinfo{zhousu@ieee.org}}
\receiveddate{}
\reviseddate{}
\Editor{}

\address[1]{School of Cyber Science and Engineering, Xi'an Jiaotong University, Xi'an 710049, China}

\begin{document}

\maketitle

\begin{abstract}
In commercial unmanned aerial vehicle (UAV) applications, one of the main restrictions is UAVs' limited battery endurance when executing persistent tasks.
With the mature of wireless power transfer (WPT) technologies, by leveraging ground vehicles mounted with WPT facilities on their proofs, we propose a mobile and collaborative recharging scheme for UAVs in an on-demand manner. Specifically, we first present a novel air-ground cooperative UAV recharging framework, where ground vehicles cooperatively share their idle wireless chargers to UAVs and a swarm of UAVs in the task area compete to get recharging services. Considering the mobility dynamics and energy competitions, we formulate an energy scheduling problem for UAVs and vehicles under practical constraints. A fair online auction-based solution with low complexity is also devised to allocate and price idle wireless chargers on vehicular proofs in real time. We rigorously prove that the proposed scheme is strategy-proof, envy-free, and produces stable allocation outcomes. The first property enforces that truthful bidding is the dominant strategy for participants, the second ensures that no user is better off by exchanging his allocation with another user when the auction ends, while the third guarantees the matching stability between UAVs and UGVs. Extensive simulations validate that the proposed scheme outperforms benchmarks in terms of energy allocation efficiency and UAV's utility.
\end{abstract}

\section{Introduction}
\label{Introduction}
The emerging unmanned aerial vehicles (UAVs) have gained significant success in various applications such as crop surveys, search and rescue, and infrastructure inspection \cite{9849496,9732222,9456851,wang2023survey}.
Thanks to their low cost, flexible deployment, and controllable maneuverability, UAVs mounted with rich onboard sensors can be fast dispatched to enable autonomous and on-demand mission execution (e.g., sensing and communication recovery) anytime and anywhere \cite{9453820,10106022,9152148}. 
However, commercial UAVs such as quadrotors generally have stringent space and weight limitations, causing inherent constrained battery endurance to support long-duration missions. For example, most mini-UAVs (powered by lithium-ion or lithium polymer batteries) only afford up to 90 minutes of endurance \cite{9743346}. Besides, in executing complex and persistent tasks, relevant compute-intensive operations with video streaming and image processing may consume a considerable amount of UAV's battery energy \cite{10106022,9696188}.
Notably, increasing UAV's battery capacity beyond a certain point can degrade its flight time due to excessive weight \cite{9420719}.
Hence, it is crucial to design effective battery recharging approaches to sustain the life cycle of a UAV flight.

A number of research efforts have been made to address the UAV's battery recharging issue, which can be mainly divided into three types: \emph{energy harvesting} \cite{9060991} from the environment (e.g., solar and wind energy), \emph{battery hotswapping} \cite{6701199} at battery swap stations, and \emph{wireless charging} \cite{9488324} using wireless chargers. In energy harvesting, the energy output of outfitted photovoltaic (PV) arrays or turbine generators can be intermittent and uncertain and highly rely on weather conditions. Besides, the additionally added size and weight on the UAV may raise difficulty in safely landing at dedicated locations.
In battery hotswapping, it generally involves human labors to replace UAV's depleted battery with a fully charged one, affecting autonomous UAV operations in inaccessible or hazardous places \cite{9420719}.
Moreover, it can incur high round-trip energy costs for frequent battery replacement operations.
With the recent breakthrough in wireless power transfer (WPT) techniques, UAVs can be conveniently charged by distributed wireless chargers in a fully automatic manner \cite{9420719,9488324,8846189,9209109}.
As reported, the commercial WPT product of Powermat company can transfer 600\,W of wireless power over a distance of up to 150\,mm to small or medium UAVs with over $90\%$ energy efficiency and high misalignment tolerance \cite{WPTDronePowermat}.
However, deploying and maintaining such static wireless chargers at large-scale task areas (e.g., survivor rescue in disaster sites) can be costly and time-consuming, especially in environmentally harsh terrains.
Besides, the solutions built on static wireless chargers usually lack feasibility and on-demand energy supply capabilities for UAVs, as well as restricting UAVs' operations within specific geographical areas.

\begin{figure}[!t]
\centering
  \includegraphics[width=8.0cm]{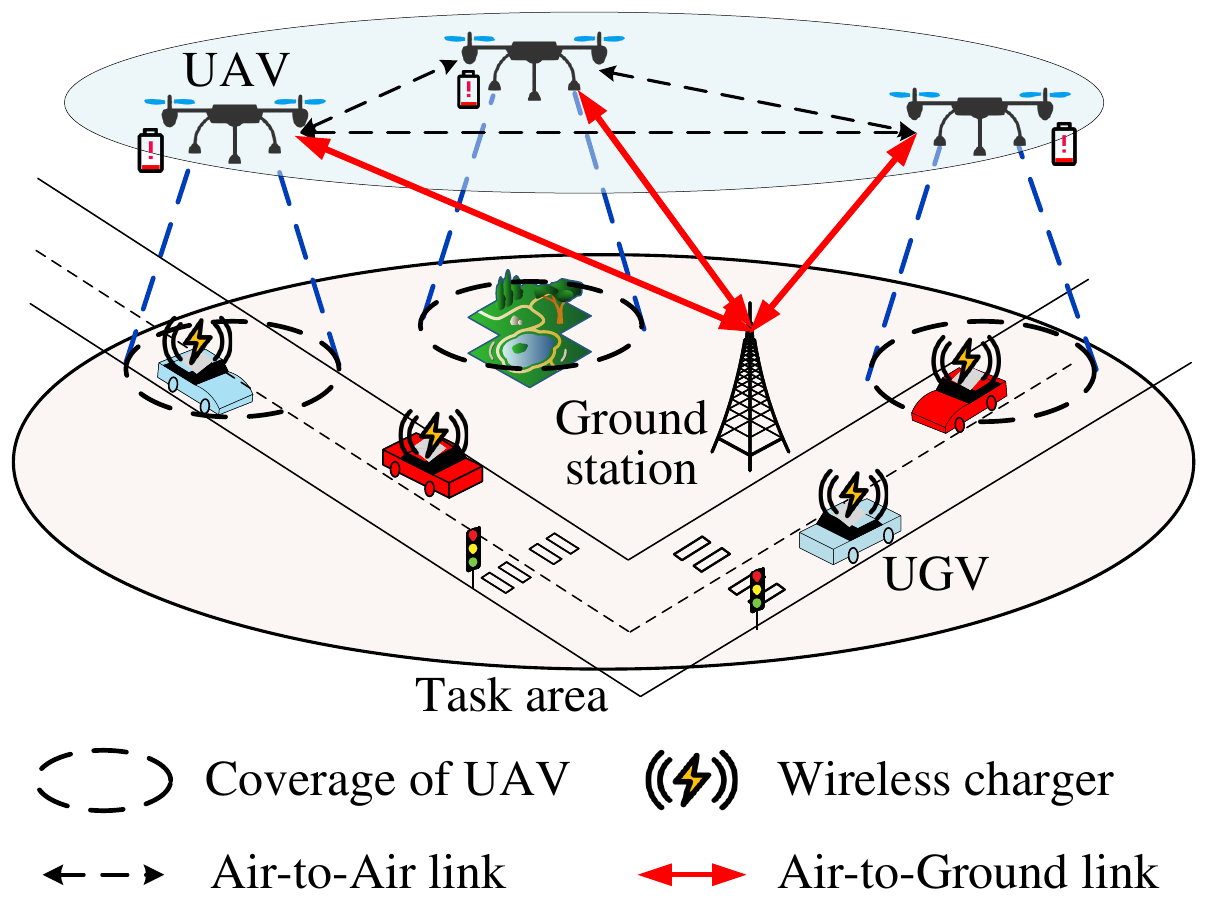}
  \caption{An example scenario of vehicle-mounted mobile wireless chargers for on-demand UAV recharging in the task area.}\label{fig:intro1}
\end{figure}

In this paper, as shown in Fig.~\ref{fig:intro1}, we focus on an on-demand and cost-effective solution by leveraging unmanned ground vehicles (UGVs) with controllable mobility, where UGVs equipped with wireless charging facilities are deployed in task areas and collaboratively offer sufficient wireless energy supply to prolong the lifetime of the UAV network.
In academia and industry, research works \cite{7535886,9058225,9462603} and companies such as Renault \cite{Renault2014} and DSraider \cite{DSraider2021} have developed such mobile and collaborative platforms for efficient UAV launching, recycling, and recharging using a special compartment in the UGV's roof.

Despite the fundamental contributions on system and protocol design of existing literature \cite{7535886,9058225,9462603,Renault2014,DSraider2021}, the double-side energy scheduling along with user fairness in UG\underline{V}-assisted \underline{w}ireless \underline{r}echargeable \underline{U}AV \underline{n}etworks (VWRUNs) are rarely studied, which motivates our work.
On one hand, compared with fixed chargers, the size of charging (also landing) pads on roofs of UGVs are comparably smaller \cite{8660495}, thereby restricting the number of concurrent charging UAVs. Moreover, in complex missions (e.g., large-scale surveillance), it usually depends on the coordination among multiple UAVs due to the limited capacity (e.g., sensing range) of a single UAV \cite{6701199,9488324}.
Consequently, in highly dynamic VWRUNs, efficient real-time charging scheduling among multiple UAVs and multiple UGVs is of necessity to motivate their energy cooperation.
On the other hand, as UAVs and UGVs are self-interested agents and mutually distrustful, they may behave strategically to maximize their gains and even perform market manipulation by diminishing the legitimate interests of others \cite{8902165}. For example, strategic UGVs may collude to overclaim their energy costs for higher payments from UAVs. Besides, the envy-freeness (i.e., no agent envies the allocation of another agent) \cite{9496271}, as an essential metric to ensure market fairness, is neglected in most of existing works. The violation of envy-freeness may raise low user willingness and acceptance, eventually lowering energy allocation efficiency.
Therefore, it remains an open and vital issue to design a real-time and envy-free charging strategy among UAVs and UGVs while preventing strategic behaviors and motivating their dynamic cooperation in VWRUNs.

To address the above issues, we adopt a market-based approach by formulating the double-side charging scheduling problem among multiple UAVs and UGVs in VWRUNs as an online sealed-bid auction, where both UAVs (i.e., buyers) and UGVs (i.e., sellers) are allowed to send their sealed bidding information (including trading time, valuation, supply/demand energy volume, etc.) anytime to the auctioneer (i.e., the ground station).
The auctioneer collects the bids within a maximum waiting time and publishes the auction outcome when the auction ends. UAVs that fail to match a desired UGV can participate the next-round auction or alternatively fly to a nearby static energy swap/charging station to replenish energy. 
Besides, our proposed scheme allows UAVs and UGVs to dynamically join and exit the auction process.
Using rigorous theoretical analysis, we prove that the proposed battery recharging auction is \emph{strategy-proof} (i.e., able to resist strategic agents) and produces \emph{envy-free} allocations (i.e., ensuring market fairness).
The main contributions of this paper are summarized as below.
\begin{itemize}
  \item We propose an on-demand and collaborative UAV recharging framework by employing idle wireless chargers mounted on mobile UGVs to replenish energy for multiple UAVs in executing long-term tasks. An optimization problem for energy scheduling is formulated in VWRUNs based on the UGV type model and UAV's state-of-charge (SoC) model under practical constraints.
  \item We devise an online auction-based approach to solve the real-time energy scheduling problem with low complexity, consisting of the winner determination phase and pricing phase. We theoretically analyze the equilibrium strategy of participants and rigorously prove its strategy-proofness, envy-freeness, and stability. 
  \item We carry out extensive simulations to evaluate the feasibility and effectiveness of the proposed scheme. Numerical results demonstrate the superiority of the proposed scheme in terms of energy allocation efficiency, UAV's utility, and social surplus, in comparison with conventional schemes.
\end{itemize}

The remainder of this paper is organized as follows. Section \ref{sec:RELATEDWORK} surveys the related literature, and Section \ref{sec:SYSTEMMODEL} introduces the system model. The detailed design of the proposed scheme is presented in Section \ref{sec:FRAMEWORK}, and its performance is evaluated in Section \ref{sec:SIMULATION} using simulations. Finally, this paper is concluded in Section \ref{sec:CONSLUSION}.

\section{Related Works}\label{sec:RELATEDWORK}
In this section, we review related literature on static/mobile wireless charging solutions and charging scheduling approaches in UAV networks.
\subsection{Static and Mobile Wireless Charging Solutions for UAVs}\label{subsec:relatedwork1}
In modern UAV applications, limited battery capacity poses a significant operational challenge to UAVs' flight durations, especially in executing large-scale persistent missions.
Compared with contact-based conductive charging techniques, the promising WPT technology offers a contact-free and fully automatic wireless charging solution for UAVs while withstanding challenging weather conditions \cite{8846189,9488324,9420719}.
Based on the transmission range, current WPT techniques for UAVs can be categorized into two types \cite{9420719}: \emph{near-field} and \emph{far-field}. The former mainly includes magnetic resonance coupling (MRC) \cite{8846189} and capacitive coupling \cite{8267698}, while the latter usually refers to as non-directive radio frequency (RF) radiation such as laser charging \cite{8995773} and WISP-reader charging \cite{9488324}.

Existing WPT-based UAV charging approaches are mainly built on static wireless charging pads located at building rooftops, power poles, cell towers, etc. For large-scale UAV missions, it highly relies on and bears the costly deployment/maintenance fee of additional wireless charging infrastructures.
In the literature, few works have attempted to design mobile wireless chargers to build a feasible and on-demand solution to sustain large-scale persistent UAV operations. Wu \emph{et al}. \cite{9058225} implemented a collaborative UAV-UGV recharging system, where the UGV equips with an object tracking camera (to automatically maneuver towards the UAV) and a Qi charger on the landing pad (to wirelessly transfer energy to the UAV after landing).
To address the misalignment issues between transmit and receiving coils in WPT-based UAV recharging systems, Rong \emph{et al}. \cite{9349168} designed an optimized coupling mechanism for UAVs with high misalignment tolerance based on the genetic algorithm. A real implementation shows that the design UAV recharging system can transfer a maximum power of 100W with the WPT efficiency of 92.41\%.
Ribeiro \emph{et al}. \cite{9462603} investigated the route planning problem for multiple mobile charging platforms (which can travel to different locations) to support long-duration UAV operations. In \cite{9462603}, the routing problem is formulated as a mixed-integer linear programming (MILP) model and solved using a genetic algorithm together with a construct-and-adjust heuristic method.
{There have been several recent studies leveraging WPT for wireless services and UAV services. Wu \emph{et al}. \cite{9718086} proposed a non-orthogonal multiple access (NOMA)-aided federated learning (FL) framework with WPT, where the base station uses WPT to recharge end devices that perform local training and data transmission in FL. A layered algorithm was also designed in \cite{9718086} to minimize FL convergence latency and overall energy consumption under practical constraints.
Shen \emph{et al}. \cite{9440683} studied a UAV-aided flexible radio resource slicing mechanism in 5G uplink radio access networks (RANs), where the joint 3D placement of UAVs and UAV-device association problem was formulated via an interference-aware graph model. In addition, a lightweight approximation algorithm and an upgraded clique method were devised in \cite{9440683} for reduced complexity.}

{However, existing works mainly focus on the system design and trajectory planning of VWRUNs, whereas the double-side energy scheduling among UAVs and UGVs along with the user fairness in the energy charging market are rarely studied.}

\subsection{WPT Charging Scheduling Methods for UAVs}\label{subsec:relatedwork2}
In the literature, there has been an increasing interest in designing {WPT-based} charging scheduling methods for UAVs.
Zhao \emph{et al}. \cite{8995773} proposed a power optimization method in a static laser charging system for a rotary-wing UAV. A non-convex optimization problem is formulated with coupling variables and practical mobility, transmission, and energy constraints, and two algorithms are devised to search the optimal strategy with guaranteed convergence to stationary solutions.
Li \emph{et al}. \cite{9488324} designed an energy-efficient charging time scheduling algorithm to turn on static wireless chargers (SWCs) in scheduled time periods with the aim to minimize SWCs' energy waste in wirelessly charging UAVs. In their scheme, UAVs' continuous flight trajectories are discretized in both temporal and spatial dimensions, and a pruning-based exhaustive method is devised for near-optimal solution searching.
Yu \emph{et al}. \cite{8460819} studied the problem of route planning for a mini-UAV to visit multiple sensing sites within its battery lifetime, where UGVs acting as mobile recharging stations can offer battery recharging services for the UAV along its tour.
Shin \emph{et al}. \cite{8660495} presented a second-price auction-based charging time scheduling method between multiple UAVs and a ground vehicle, where the ground vehicle serves as a mobile charging station and the auctioneer. In their auction, multiple UAVs bid for the vehicle's charging time slots, and the UAV with the highest valuation is assigned as the winner.

{One can observe that existing works on UAV charging mainly focus on the one-to-one charging pattern \cite{8995773} or many-to-one charging pattern \cite{9488324,8660495}, which is not applicable in our considered scenario with multiple UAVs and multiple UGVs. Besides, due to the high dynamics of VWRUNs and potential strategic entities during charging scheduling, real-time and fair energy scheduling should be enforced, which is rarely studied.} Distinguished from existing works, we design an envy-free online auction mechanism for efficient many-to-many charging scheduling among multiple UAVs and UGVs with consideration of their mobility dynamics, double-side competition, and practical constraints.

\section{System Model}\label{sec:SYSTEMMODEL}
In this section, we first elaborate on the system model including the network model (in Sect.~\ref{subsec:networkmodel}) and UAV energy consumption and wireless charging model (in Sect.~\ref{subsec:WPTmodel}). {The notations used in this paper is summarized in Table~\ref{table0}.}

\begin{table}[!t]
{
\caption{{Summary of Notations}}\label{table0}\centering
\resizebox{1.02\linewidth}{!}{
\begin{tabular}{|c|l|}
\hline 
\textbf{Notation} & \textbf{Description} \\   \hline
$\mathcal{I}$&Set of UAVs to be recharged. \\
$\mathcal{J}$&Set of UGVs with WPT facilities. \\
$\Lambda$&The GS that coordinates UAVs and UGVs. \\
${\mathcal{I}}'$&Set of low-battery UAVs with $s_i[t] \leq s_{\min}$. \\
${\mathcal{J}}'$&Set of UGVs with idle WPT facilities in time window $\tau$. \\
${\mathcal{W}}$&Winner set of UAVs in the auction. \\
$\mathbb{A}$&UAV recharging auction. \\
$T$&Finite time horizon containing $N$ time slots. \\
$\Delta_t$&Duration of each time slot. \\
$\mathbf{l}_i[t]$&Instant 3D location of UAV $i$ at $t$-th time slot. \\
$R_i$&Radius of sensing spot of UAV $i$. \\
$\theta_i$&Maximum detection angle of UAV $i$'s sensor. \\
$z_{\max}$&Maximum flight altitude of UAV $i$. \\
$v_i$&Flying velocity of UAV. \\
${v_{\max}}$&Maximum velocity of UAV. \\
$s_i[t]$&SoC of UAV $i$'s battery at $t$-th time slot. \\
$C_i$&Battery capacity of UAV $i$. \\
$s_{\min}$&Minimum reserved battery energy. \\
$s_i^{\mathrm{sat}}$&Satisfactory SoC level of UAV $i$. \\
$\rho_i[t]$&Charging urgency of UAV $i$. \\
$C_j$&Total wireless energy supply of UGV $j$. \\
$s_j[t]$&Remaining wireless energy supply of UGV $j$. \\
$q_j[t]$&QoRS of UGV $j$ at $t$-th time slot. \\
$P_i^{\mathrm{fly}}$&Flying power of UAV $i$. \\
$P_i^{\mathrm{hov}}$&Hovering power of UAV $i$. \\
$P_i^{\mathrm{d}}$&Power of UAV $i$ in the descending process. \\
$P_i^{\mathrm{a}}$&Power of UAV $i$ in the ascending process. \\
$P_j^{\mathrm{e}}$&Wireless power transferred by the UGV $j$. \\
$\alpha _i^u$&State variable of UAV $i$. \\
$\eta_j$&Wireless power efficiency of UGV $j$. \\
${\boldsymbol{{b}}}$&Bid profile of all UAVs. \\
$\Phi_i[t]$&Valuation or reserve price of UAV $i$. \\
$\mathcal{U}_i$&Utility function of UAV $i$. \\
$\beta_{i,j}$&Binary allocation outcome. \\
$\bar{\Phi}_i$&Average valuation of UAV $i$ during time window $\tau$. \\
$p_j$&Payment to UGV $j$. \\
$\mathcal{U}_j$&Utility function of UGV $j$. \\
$\mathcal{S}$&Social surplus of involved entities. \\
$g(i)$&Identity of allocated UGV to UAV $i$ in the auction. \\
$\tau$&Time window of the auction. \\
\hline 
\end{tabular} } }
\end{table}

\subsection{Network Model}\label{subsec:networkmodel}
As depicted in Fig.~\ref{fig:intro1}, we consider a typical scenario of VWRUN in a given investigated area, which mainly consists of a swarm of $I$ UAVs, a fleet of $J$ UGVs, and a ground station (GS).

\emph{\textbf{UAVs}.} Due to the limited sensing coverage and energy supply of a single UAV, a swarm of UAVs, denoted by the set $\mathcal{I} = \{1,\cdots, i,\cdots,I\}$, are dispatched to collaboratively execute a common mission (e.g., air quality monitoring and geographic surveying) in the given task area \cite{9035635}. The sensing spot of UAV $i \in \mathcal{I}$ is denoted as a circle with center $\left(x_i,y_i,0\right)$ and radius $R_i$. UAVs can communicate with each other using air-to-air (A2A) communications \cite{9599638}.
Let $T$ denote the finite time horizon, which is evenly divided into $N$ time slots with duration $\Delta_t$ \cite{10106022}, i.e., $T = N\cdot \Delta_t$. The instant 3D location of UAV $i \in \mathcal{I}$ at $t$-th time slot ($1\le t \le T$) is denoted by $\mathbf{l}_i[t] = (x_i[t],y_i[t],z_i[t])$, where $(x_i[t],y_i[t])$ is its instant horizontal coordinate. 
The instant altitude $z_i[t]$ of UAV $i$ satisfies
\begin{equation}\label{eq3-1}
R_i \cot(\theta_i) \le z_{i}[t] \le z_{\max},
\end{equation}
where $\theta_i$ is the maximum detection angle of UAV $i$'s sensor and $z_{\max}$ is its maximum flight altitude. Besides, $\left||\mathbf{l}_i[t+1] - \mathbf{l}_i[t]|\right| \leq {v_{\max}} \Delta_t$, $\forall 1\le t \le T-1$, where ${v_{\max}}$ is the maximum velocity of UAV.

The state-of-charge (SoC) of UAV $i$'s on-board battery at $t$-th time slot is $s_i[t]$, which satisfies $s_{\min}\leq s_i[t] \leq C_i$. Here, $C_i$ is the battery capacity of UAV~$i$ and $s_{\min}$ is the minimum reserved battery energy to prolong the battery lifetime \cite{9632356}. For UAV $i$, when its remaining battery SoC $s_i[t]$ is below the alert level $s_{\min}$, it leaves its sensing spot for recharging and another UAV can cooperatively replace this low-battery UAV $i$ at the target sensing spot to offer uninterrupted sensing service \cite{8648453}. The charging urgency of each UAV $i$ is computed as
\begin{equation}\label{eq3-1-2}
\rho_i[t] = 1 - \frac{s_i[t] - s_{\min}}{C_i}, \,\mathrm{and}~ \rho_i[t] \in [0,1].
\end{equation}

\emph{\textbf{UGVs}.} A fleet of UGVs equipped with wireless charging facilities on the vehicular roofs are deployed in the investigated area to collaboratively offer on-demand wireless energy supply for low-battery UAVs \cite{Renault2014,DSraider2021}. The set of UGVs is denoted as $\mathcal{J} = \{1,\cdots, j,\cdots,J\}$. UGVs are smart vehicles integrated with various advanced sensors to allow self-driving to the rendezvous and perform automatically UAV tracking, launching, and recharging operations on the charging pad on vehicular roofs. Let $C_j$ and $s_j[t]$ be the total/remaining wireless energy supply of UAV recharging {of UGV $j$}, respectively. It is assumed that $s_j[t] \geq \max\{s_i^{\mathrm{sat}} - s_i[t], \forall i \in \mathcal{N}\}$.
Besides, UGVs generally have diverse quality of recharging service (QoRS), which is affected by various factors such as the wireless charging rate and the driving distance to the task area. Let $q_j[t]$ denote the QoRS of UGV $j$ at $t$-th time slot. A higher QoRS indicates the higher charging preference of UAVs.

\emph{\textbf{GS}.} In VWRUN, the aerial UAV subnetwork and the ground vehicular subnetwork are coordinated by the GS (denoted by $\Lambda$) \cite{7317860}. The GS is located at a micro base station and can perform flight planning, flying control, and task assignment for UAVs via ground-to-air (G2A) links. Moreover, after receiving recharging requests from UAVs, the GS can schedule the UGVs with idle wireless chargers in its communication range via infrastructure-to-vehicle (I2V) links \cite{9631953} to offer on-demand recharging services.

\subsection{UAV Energy Consumption and Wireless Charging Model}\label{subsec:WPTmodel}
To avoid collisions and save energy in the flight, UAVs need to horizontally fly over the task area and hover above the assigned task spot to perform sensing missions \cite{10155496}. In energy recharging process, for simplicity, each UAV $i$ vertically descends to the target UGV's proof and vertically ascends to a preset altitude after reaching the satisfactory SoC level $s_i^{\mathrm{sat}}$ \cite{8758340}. According to \cite{7991310}, the required flying power at a constant speed $v_i$ can be approximately attained as:
\begin{equation}\label{eq3-2}
P_i^{\mathrm{fly}} (v_i) = \kappa_1 v_i^3 + (\kappa_2 + \kappa_3) \Psi^{3/2},
\end{equation}
where $\kappa_1, \kappa_2, \kappa_3$ are constant UAV-related factors, $\Psi$ is the thrust of UAV \cite{8758340}. The hovering power of UAV~$i$ is $P_i^{\mathrm{hov}} = (\kappa_2 + \kappa_3)(m g)^{3/2}$, where $m$ is UAV's mass and $g=9.8\,\mathrm{m/s^2}$. Given the fixed descending speed $v_i^{\mathrm{d}}$ and ascending velocity $v_i^{\mathrm{a}}$, the required power of UAV in the descending process and ascending process can be separately expressed as \cite{8758340,7991310}:
\begin{equation}\label{eq3-3}
P_i^{\mathrm{d}} (v_i^{\mathrm{d}}) \!=\! \epsilon_1 m g \left[\sqrt{\frac{(v_i^{\mathrm{d}})^2}{4} \!+\! \frac{m g}{(\epsilon_2)^2}} - \frac{v_i^{\mathrm{d}}}{2} \right] + \kappa_3(m g)^{3/2}\!,\!
\end{equation}
\begin{equation}\label{eq3-4}
P_i^{\mathrm{a}} (v_i^{\mathrm{a}}) \!=\! \epsilon_1 m g \left[\sqrt{\frac{(v_i^{\mathrm{a}})^2}{4} \!+\! \frac{m g}{(\epsilon_2)^2}} + \frac{v_i^{\mathrm{a}}}{2} \right] + \kappa_3(m g)^{3/2}\!,\!
\end{equation}
where $\epsilon_1, \epsilon_2$ are constant UAV-related factors.

Let $P_j^{\mathrm{e}}$ denote the wireless power transferred by the UGV $j$. Then, the battery dynamics of UAV $i$ can be described as a linear model, i.e.,
\begin{equation}\label{eq3-5}
s_i[t+1] \!=\! s_i[t] +
\left[ \begin{matrix}{}
	\alpha _i^1\!&\!	\alpha _i^2\!&\!  \alpha _i^3\!&\!	\alpha _i^4\!&\!	\alpha _i^5\\
\end{matrix} \right] \!\times\! \left[ \begin{array}{c}
	- \eta_i P_i^{\mathrm{fly}}\\
	- \eta_i P_i^{\mathrm{hov}}\\
	- \eta_i P_i^{\mathrm{d}}\\
	- \eta_i P_i^{\mathrm{d}}\\
	\eta_i \eta_j P_j^{\mathrm{e}}\\
\end{array} \right]
\!,\!
\end{equation}
where $\alpha _i^u =\{0,1\},u=\{1,\cdots,5\}$ are binary variables, denoting the state of UAV $i$. $\eta_j$ is the wireless power efficiency of UGV $j$. Here, $\alpha _i^u=1$ means that UAV $i$ is in the corresponding state (i.e., horizontally flying, hovering, vertically descending, vertically ascending, or wireless charging); otherwise, $\alpha _i^u=0$.

\begin{figure}[!t]\setlength{\abovecaptionskip}{0.1cm}
\centering
  \includegraphics[width=9.0cm,height=5.8cm]{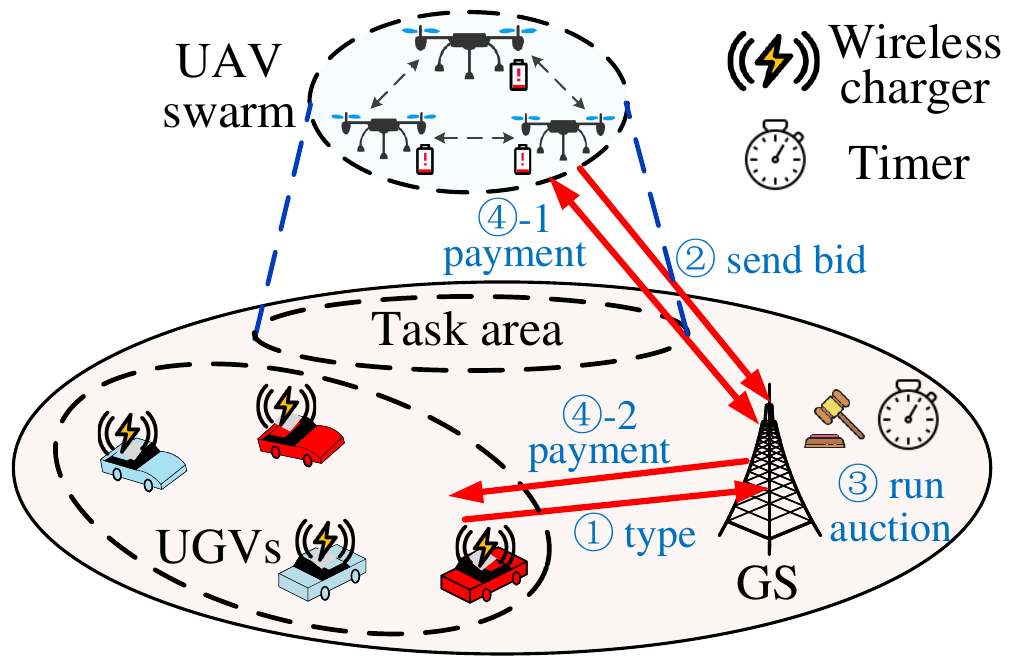}
  \caption{Illustration of online auction-based charging scheduling process among UAVs and UGVs (\ding{172}: Based on the announced auction duration, UGVs with idle wireless chargers compute their types and send to the GS; \ding{173}: UAVs submit sealed bids to the GS; \ding{174}: the GS runs the auction and determines the winners and payments; \ding{174}: each winning UAV delivers its payments to the corresponding UGV via the GS for battery recharging.)}\label{fig:Auction}
\end{figure}

\section{The Proposed Scheme}\label{sec:FRAMEWORK}
This section first formulates the online charging scheduling and pricing problem for UAVs and UGVs in VWRUNs (in Sect.~\ref{subsec:scheme1}). Then, an auction-based solution with strategy-proofness and envy-freeness is designed, followed by the theoretical analysis of its properties (in Sect.~\ref{subsec:scheme2}).

\subsection{Online Charging Scheduling and Pricing (OCSP) Problem}\label{subsec:scheme1}
As shown in Fig.~\ref{fig:Auction}, the auction-based UAV charging scheduling process is carried out by the GS in an \emph{online} manner, where UAVs are allowed to bid at anytime and the bid collection phase finishes until a maximum waiting time $\tau$ elapses.
Let ${\boldsymbol{{b}}} = \left(b_1,\cdots,b_i,\cdots,b_{I(\tau)} \right)$ denote the bid profile of all UAVs in ${\mathcal{I}}(\tau) \subseteq \mathcal{I}$. Here, ${\mathcal{I}}' = {\mathcal{I}}(\tau)$ is the set of low-battery UAVs with $s_i[t] \leq s_{\min}$, $\forall i \in \mathcal{N}, \forall t \in \tau$. ${\boldsymbol{{b}}_{-i}}$ is the bid profile of other UAVs except UAV $i$.

In the auction, the valuation (i.e., reserve price) of UAV $i$ in a recharging service is associated to its energy state and charging urgency, i.e., $\Phi_i[t] = \Phi_i(\rho_i[t])$. Generally, the higher charging urgency, the larger valuation. Besides, the higher charging urgency, the larger marginal valuation. Hence, $\frac{\mathrm{d} \Phi_i(\rho_i[t])}{\mathrm{d} \rho_i[t]} \!>\! 0$, $\frac{\mathrm{d}^2 \Phi_i(\rho_i[t])}{\mathrm{d} \rho_i[t]^2} \!\geq\! 0$. In the following, we define utility functions of UAVs and UGVs, as well as the social surplus.

\begin{definition}[UAV Utility]\label{definition1}
The utility function of UAV $i \in \mathcal{N}$ is the revenue minuses its payment:
\begin{equation}\label{eq4-1}
\mathcal{U}_i = \left\{\begin{array}{cl}
\sum\nolimits_{j\in \mathcal{J}}{\beta_{i,j} \left[q_j \bar{\Phi}_i - p_j( {\boldsymbol{{b}}})\right]}, & i \in  {\mathcal{I}}',\\
0, & i \!\in\!  {\mathcal{I}}\backslash{\mathcal{I}}'.
\end{array}\right.
\end{equation}
\end{definition}
\emph{Remark.} In Eq. (\ref{eq4-1}), the binary variable $\beta_{i,j}=\{0,1\}$ indicates the allocation outcome, where $\beta_{i,j}=1$ if UAV $i$ is allocated to get charged at UGV $j$. Otherwise, $\beta_{i,j}=0$. $q_j$ is the QoRS of UGV $j$, which is assumed to remain unchanged during the time window $\tau$. $\bar{\Phi}_i$ is the average valuation of UAV $i$ during the time window $\tau$, which is computed as $\bar{\Phi}_i = \lfloor \frac{\tau}{t} \rfloor^{-1} \cdot \sum_{t\in \tau}{\Phi}_i[t]$. $p_j({\boldsymbol{{b}}})$ is the payment to UGV $j$.
\begin{definition}[UGV Utility]\label{definition2}
The utility function of UGV $j \in \mathcal{J}$ is associated with its payment, i.e.,
\begin{equation}\label{eq4-2}
\mathcal{U}_j = \left\{\begin{array}{cl}
\sum\nolimits_{i\in \mathcal{I}'} {\beta_{i,j} p_j({\boldsymbol{{b}}})}, & j \in  {\mathcal{J}}',\\
0, & j \in  {\mathcal{J}}\backslash{\mathcal{J}}'.
\end{array}\right.
\end{equation}
\end{definition}
\emph{Remark.} In Eq. (\ref{eq4-2}), ${\mathcal{J}}'={\mathcal{J}}(\tau)$ denotes the set of UGVs with idle WPT facilities during time window $\tau$, where ${\mathcal{J}}(\tau)\subseteq \mathcal{J}$. 
\begin{definition}[Social Surplus]\label{definition3}
The social surplus is defined as the overall utility of involved entities \cite{8902165}, i.e.,
\begin{equation}\label{eq4-3}
\mathcal{S}  = \sum\limits_{i\in \mathcal{I}} {\mathcal{U}_i} + \sum\limits_{j\in \mathcal{J}} {\mathcal{U}_j} = \sum\limits_{i\in \mathcal{I}'}{\sum\limits_{j\in \mathcal{J}'} {{\beta_{i,j} q_j \bar{\Phi}_i }}}.
\end{equation}
\end{definition}
Besides, the UAV recharging auction should be strategy-proof and envy-free to prevent strategic entities and ensure market fairness, whose formal definitions are given as below.
\begin{definition}[Strategy Proofness]\label{definition4}
The UAV recharging auction $\mathbb{A}$ satisfies strategy proofness if the following two properties hold \cite{10155496}:
\begin{itemize}
  \item (i) Individual rationality (IR): both UAVs and UGVs acquire non-negative utilities in the auction, i.e., $\mathcal{U}_i\geq 0, \forall i \in \mathcal{I}'$ and $\mathcal{U}_j\geq0, \forall j \in \mathcal{J}$.
  \item (ii) Individual compatibility (IC): each UAV can obtain its maximum utility when truthfully choosing its bid strategy, i.e., $\mathcal{U}_i(\bar{\Phi}_i,\boldsymbol{b}_{-i}) \geq \mathcal{U}_i(b_i',\boldsymbol{b}_{-i}), \forall b_i' \neq \bar{\Phi}_i, i \in\mathcal{I}'$.
\end{itemize}
\end{definition}
\begin{definition}[Envy Freeness]\label{definition5}
The UAV recharging auction $\mathbb{A}$ is energy-free if no UAV is happier to exchange its allocation with another UAV to improve its utility when the auction ends \cite{9496271}, i.e.,
\begin{equation}\label{eq4-4}
\mathcal{U}_i(g(i),p_{g(i)}) \geq \mathcal{U}_i (g(k),p_{g(k)}), k\neq i, \forall i,k\in \mathcal{I}',
\end{equation}
where $g(i)$ is the identity of allocated UGV to UAV $i$ and $p_{g(i)}$ is the corresponding payment of UAV $i\in \mathcal{I}'$.
\end{definition}

In VWRUNs, the online charging scheduling and pricing (OCSP) problem is to maximize the social surplus while meeting practical constraints, i.e.,
\begin{equation}\label{eq4-5}
\mathbf{P}1:
\mathop {\max }\limits_{{\boldsymbol{{\beta}}},\,\boldsymbol{p}} \sum\limits_{i\in \mathcal{I}'}{\sum\limits_{j\in \mathcal{J}'} {{\beta_{i,j} q_j \bar{\Phi}_i }}},
\end{equation}
\begin{numcases}{{\rm{s.t.}}}	
\sum\nolimits_{i \in {\mathcal{I}'}} {{\beta_{i,j}}} \le 1 \label{eq:cons1} \hfill \\
\sum\nolimits_{j \in \mathcal{J}'} {{\beta_{i,j}}}  \le 1 \label{eq:cons2} \hfill \\
{\beta_{i,j}} = \left\{0,1\right\},{p_j({\boldsymbol{{b}}})} \ge 0, \forall i \!\in\! {\mathcal{I}'},\forall j \!\in\! \mathcal{J}' \label{eq:cons3} \hfill \\
\sum\nolimits_{u=1}^{5} \alpha_i^u = 1, \forall i \in \mathcal{I}' \label{eq:cons4} \hfill \\
\mathcal{U}_i\geq 0, \mathcal{U}_j\geq 0, \forall i \in \mathcal{I}' ,\forall j \in \mathcal{J}' \label{eq:cons5} \hfill \\
\mathcal{U}_i(\bar{\Phi}_i,\boldsymbol{b}_{-i}) \!\geq\! \mathcal{U}_i(b_i',\boldsymbol{b}_{-i}), \forall b_i' \!\neq\! \bar{\Phi}_i, i \!\in\! \mathcal{I}' \label{eq:cons6} \hfill \\
\mathcal{U}_i(g(i),p_{g(i)}) \!\geq\! \mathcal{U}_i (g(k),p_{g(k)}), \forall k \neq i. \label{eq:cons7} \hfill 
\end{numcases}

\emph{Remark.} In $\mathbf{P}1$, ${\boldsymbol{{\beta}}} = [{\beta_{i,j}}]_{I(\tau) \times J(\tau)}$ and ${\boldsymbol{p}}=(p_1,\cdots,p_{J(\tau)})$ are decision variables. Constraint (\ref{eq:cons1}) means that a UGV can only offer charging service for at most one UAV during $\tau$. Constraint (\ref{eq:cons2}) implies that a UAV can only recharge at most one UGV during $\tau$. Constraint (\ref{eq:cons4}) is the state constraint of UAV $i$. Constraint (\ref{eq:cons5}) is the IR constraint, and constraint (\ref{eq:cons5}) is the IC constraint.
Both constraints (\ref{eq:cons5})--(\ref{eq:cons6}) refer to the strategy proofness, and constraint (\ref{eq:cons7}) corresponds to the envy freeness.

\begin{theorem}\label{theorem0}
The OCSP problem $\mathbf{P}1$ is NP-hard.
\end{theorem}
\begin{proof}
Based on \cite{7293173}, the relaxed version of problem $\mathbf{P}1$ with constraints (\ref{eq:cons1})--(\ref{eq:cons4}) and the constant payment ${\boldsymbol{p}}$ is a typical set cover problem and is proved to be NP-hard. As such, the raw OCSP problem $\mathbf{P}1$ is NP-hard. Theorem~\ref{theorem0} is proved.
\end{proof}

\subsection{Strategy-Proof and Envy-Free Auction Mechanism}\label{subsec:scheme2}
As the OCSP problem $\mathbf{P}1$ is NP-hard, in this subsection, we devise a practical heuristic auction mechanism to derive its near-optimal solution with polynomial complexity, while satisfying strategy proofness and envy freeness.
The key phases in our proposed online auction mechanism is presented {in Algorithm~\ref{Algorithm1}.} 

\begin{algorithm}[t!]\begin{footnotesize}
   \caption{\small{{Online Strategy-Proof and Envy-Free Auction Algorithm for UAVs and UGVs in VWRUNs}}}\label{Algorithm1}
   {\begin{algorithmic}[1]
   \STATE \emph{\textbf{{Phase 1}} (Type Evaluation for UAVs and UGVs):}
   \STATE Initialize the time window $\tau$, UAV type (i.e., charging urgency) $\rho_i[t]$, UGV type (i.e., QoRS) $q_j[t]$, the set of UAVs with charging desires (i.e., $\mathcal{I}'$), and the set of UGVs with idle charging facilities (i.e., $\mathcal{J}'$). Then, UAV's average valuation $\bar{\Phi}_i$, $\forall i \in \mathcal{I}'$ can be calculated.
  \STATE Initialize all elements in ${\boldsymbol{{\beta}}}$ and ${\boldsymbol{{p}}}$ with $0$.
  \STATE Initialize $\mathcal{W}=\emptyset$, $\mathcal{L}=\mathcal{I}'$, and $\mathcal{D}=\mathcal{J}'$.
 \STATE \emph{\textbf{{Phase 2}} (Biding and Allocation):}
 \STATE Each UAV $i\in {\mathcal{I}'}$ determines its sealed-bid strategy $b_i$ based on its valuation and sends it to the GS.
       \STATE The GS obtains the latest types of UGVs in $\mathcal{J}'$, and re-sorts the UGVs in descending order of their types, i.e., $q_1 \geq q_2 \geq \cdots \geq q_{J(\tau)}>0$. The sorted set of $\mathcal{J}'$ is denoted by ${\mathcal{J}''}$.
       \STATE UAVs in set ${\mathcal{I}'}$ are sorted in descending order of their valuations, i.e., $\bar{\Phi}_1 \geq \bar{\Phi}_2 \geq \cdots \geq \bar{\Phi}_{I(\tau)}\geq 0$. The new sorted set is denoted as ${\mathcal{I}''}$.
       \STATE The GS sorts the UAVs in ${\mathcal{I}''}$ in descending order of their submitted bids, and allocates the WPT facilities mounted on UGVs in ${\mathcal{J}''}$ to the UAVs starting from the maximum UGV type. The maximum type of UGV is allocated to the UAV with the largest bid (i.e., $g(1)$), the second highest type of UGV to $g(2)$, and so on, down to the type of UGV $K$, where $K = \min\{I(\tau),J(\tau)\}$, $I(\tau) =|{\mathcal{I}}'|$, and $J(\tau) =|{\mathcal{J}}'|$.
       \STATE Update $\mathcal{W} = \mathcal{W}\cup \{i|i\leq K, \forall i \in {\mathcal{I}''}\}$, $\mathcal{L} = \mathcal{L}\backslash \{i|i\leq K, \forall i \in {\mathcal{I}''}\}$, $\mathcal{D} = \mathcal{D}\cup \{j|j\leq K, \forall j \in {\mathcal{J}''}\}$.
 \STATE \emph{\textbf{{Phase 3}} (Payment Determination):}
   \STATE The payment of the last UAV $g(K)\in{\mathcal{W}}$ is\vspace{-4mm}
      \begin{equation}\label{eq4-6}
        p_{g(K)} \!=\! \left\{\begin{array}{cl}
        q_{J(\tau)}b_{g(J(\tau)+1)}, & J(\tau) \!<\! I(\tau),\\
        0, & otherwise.
        \end{array}\right.
      \end{equation}\vspace{-5mm}
      \STATE The payment of UAV $g(j)\in{\mathcal{W}}$ with $j<\min\{I(\tau),J(\tau)\}$ that wins to charge at UGV $j\in{\mathcal{J}''}$ is determined based on the negative externality that it imposes on others, i.e.,\vspace{-3mm}
      \begin{align}\label{eq4-7}
         &p_{g(j)}({\boldsymbol{{b}}})= \mathcal{S}_{{\mathcal{I}''}\backslash\{g(j)\}}^{\mathcal{J}''} - \mathcal{S}_{{\mathcal{I}''}\backslash\{g(j)\}}^{\mathcal{J}''\backslash\{j\}} \nonumber \\
         &=\Big[ \sum\nolimits_{k=1}^{j-1}{q_k b_{g(k)}} + \sum\nolimits_{k=j+1}^{K}{q_{k-1} b_{g(k)}} \Big]\nonumber \\
         & ~~~~- \Big[ \sum\nolimits_{k=1}^{j-1}{q_k b_{g(k)}} + \sum\nolimits_{k=j+1}^{K}{q_{k} b_{g(k)}} \Big]\nonumber \\
         &= \sum\nolimits_{k=j}^{K-1}{(q_k - q_{k+1}) b_{g(k+1)}}.
      \end{align}\vspace{-5mm}
   \STATE \emph{\textbf{{Phase 4}} (Wireless Charging of Winning UAVs):}
\STATE Each winner UAV $i=g(j)$ and its allocated UGV $j$ drive to their jointly produced rendezvous for battery recharging, and the winner UAV $i=g(j)$ pays the corresponding payment to the UGV $j$.
\STATE For loser UAVs, they can participate in the next-round auction, together with the new UAVs with charging desires, to compete for battery charging. Alternatively, losing UAVs can also fly to a nearby static swap/charging station to replenish energy.
\end{algorithmic}}\end{footnotesize}
\end{algorithm}

{Specifically, phase 1 (lines 1--4) evaluates the type information of UAVs and UGVs in the time window $\tau$; phase 2 (lines 5--10) determines the bidding strategy of each UAV and the auction winners; phase 3 (lines 11--13) determines the payment of each UAV in the winner set ${\mathcal{W}}$; and phase 4 (lines 14--16) performs wireless charging for each pair of matched UAV and UGV in the winner set and enforces the financial settlement.}

\emph{Remark.} In Eq.~(\ref{eq4-7}), $\mathcal{S}_{{\mathcal{I}''}\backslash\{g(j)\}}^{\mathcal{J}''}$ means the social surplus when UAV $i$ leaves the auction, and $\mathcal{S}_{{\mathcal{I}''}\backslash\{g(j)\}}^{\mathcal{J}''\backslash\{j\}}$ is the actual social surplus when UAV $i$ participates in the auction without UGV $j$. Via recursive operations, the payment for UGV $j<\min\{I(\tau),J(\tau)\}$ can be rewritten as:
\begin{equation}
p_{g(j)}({\boldsymbol{{b}}}) = (q_j - q_{j+1})b_{g(k+1)} + p_{g(j+1)}({\boldsymbol{{b}}}).
\end{equation}

In the following, we analyze the desirable properties of the proposed auction mechanism in terms of strategy-proofness, envy-freeness, allocation stability, and computational complexity in the following theorems and corollaries.

\begin{theorem}\label{theorem1}
In the proposed auction mechanism, participants always attain non-negative utilities; and the truth-telling bidding strategy is the dominant equilibria for all participating UAVs, i.e., $\mathbb{A}$ is strategy-proof.
\end{theorem}
\begin{proof}
According to Definition~\ref{definition4}, it suffices to prove that both IR and IC hold in $\mathbb{A}$. We first prove the IR. Obviously, according to Eqs. (\ref{eq4-2})--(\ref{eq4-3}), the utilities of UAVs and UGVs equal to zero if they do not participate in the auction $\mathbb{A}$. For the participating UGVs, as the payments to them are always non-negative, their utilities are no less than zero. For any participating UAV $i=g(j)$, its utility can be reformulated as:
\begin{align}
 \mathcal{U}_{g(j)}&={q_j b_{g(j)}} + \mathcal{S}_{{\mathcal{I}''}\backslash\{g(j)\}}^{\mathcal{J}''\backslash\{j\}} - \mathcal{S}_{{\mathcal{I}''}\backslash\{g(j)\}}^{\mathcal{J}''} \nonumber \\
 &= \sum\nolimits_{l=1}^{|\mathcal{J}''|}{q_l b_{g(l)}} - \mathcal{S}_{{\mathcal{I}''}\backslash\{g(j)\}}^{\mathcal{J}''} \nonumber \\
 &= \sum\nolimits_{l=j}^{|\mathcal{J}''|-1}{q_l (b_{g(l)} - b_{g(l+1)})} + {q_{|\mathcal{J}''|} b_{g(|\mathcal{J}''|)}} \nonumber \\
 &\geq 0.
\end{align}
Thereby, for all participating UAVs and UGVs, their utilities are always non-negative. Next, we prove the IC. Notably, the payment decision strategy in our proposed auction mechanism follows the Vickrey–Clarke–Groves (VCG) mechanism. As truth-telling is a well-known property of the VCG mechanism \cite{Benjamin2007Internet}, our auction $\mathbb{A}$ also satisfies the truth-telling property (i.e., IC). Theorem~\ref{theorem1} is proved.
\end{proof}

{\emph{Remark.} Theorem~2 shows that our proposed auction algorithm can resist strategic UAVs/UGVs and prevent market manipulation in practical energy recharging services by enforcing IC constraints. Besides, as IR constraints are satisfied, individual UAVs/UGVs can be motivated to join the energy recharging system to gain benefits.
}

\begin{theorem}\label{theorem2}
The proposed auction $\mathbb{A}$ is envy-free, if (i) $b_{g(j)} \in \left[\bar{\Phi}_{g(j)}, \bar{\Phi}_{g(j-1)} \right]$, $1<j\leq K$, and (ii) ${g(j)} = j$, $1\leq j\leq K$ hold.
\end{theorem}
\begin{proof}
According to Definition~\ref{definition5}, it suffices to prove that any UAV ${g(j)}\in \mathcal{W}$ is indifferent between charging at UGV $j$ at price $p_{g(j)}$ and charging at UGV $l$ at price $p_{g(l)}$, where $j\ne l$. Without loss of generality, we consider the following two cases.

\underline{Case 1}: $j>l$. In this case, as $b_{g(j)} \geq \bar{\Phi}_{g(j)}$ and $q_j\leq q_l$, we have $\left( q_j-q_l \right) \bar{\Phi} _{g\left( j \right)} \geq \left( q_j-q_l \right) b _{g\left( j \right)}$ and $\left( q_{j-2}-q_{j-1} \right)b _{g\left( j-1 \right)} \geq \left( q_{j-2}-q_{j-1} \right)b _{g\left( j \right)}$. The utility difference of UAV ${g(j)}$ between charging at UGV $j$ and charging at UGV $l$ is:
\begin{equation}\label{eq-p2-1}
\begin{aligned}
&\Delta\mathcal{U}_{j,l} = q_j\bar{\Phi} _{g\left( j \right)}-p_{g\left( j \right)}-\left[ q_l\bar{\Phi} _{g\left( j \right)}-p_{g\left( l \right)} \right]\\
&=\sum_{k=l}^{j-1}{\left( q_k-q_{k+1} \right) b_{g\left( k+1 \right)}}+\left( q_j-q_l \right) \bar{\Phi} _{g\left( j \right)}\\
&\ge \sum_{k=l}^{j-1}{\left( q_k-q_{k+1} \right) b_{g\left( k+1 \right)}}+\left( q_j-q_l \right) b_{g\left( j \right)}.
\end{aligned}
\end{equation}
If $j = l + 1$, the following inequality holds:
\begin{equation}
\Delta\mathcal{U}_{j,l}\geq \left( q_{j-1}-q_j \right) b_{g( j )} + \left( q_j-q_{j-1} \right) b_{g( j )}\geq 0.
\end{equation}
If $j > l + 1$, via recursive operations, we have
\begin{equation}\label{eq-p2-2}
\begin{aligned}
&\sum\nolimits_{k=l}^{j-1}{\left( q_k-q_{k+1} \right) b_{g\left( k+1 \right)}}+\left( q_j-q_l \right) b_{g\left( j \right)} \\
&\ge \sum_{k=l}^{j-2}{\left( q_k-q_{k+1} \right) b_{g\left( k+1 \right)}}+\left( q_{j-1}-q_l \right) b_{g\left( j \right)}\\
&\ge \cdots \ge \left( q_l-q_l \right) b_{g\left( j \right)}=0.
\end{aligned}
\end{equation}
Hence, it can be concluded that $\Delta\mathcal{U}_{j,l}\geq 0$.

\underline{Case 2}: $j<l$. In this case, as $b_{g(j)} \leq \bar{\Phi}_{g(j-1)}$ and $q_j\geq q_l$, we have $\left( q_j-q_{l} \right) \bar{\Phi}_{g(j)}\geq \left( q_j-q_{l} \right) b_{g\left( j+1 \right)}$ and $\left( q_l-q_{l-1} \right) b_{g\left( l \right)}\geq \left( q_l-q_{l-1} \right) b_{g\left( l-1 \right)}$.
Then, we can obtain:
\begin{equation}\label{eq-p2-3}
\begin{aligned}
&\Delta\mathcal{U}_{j,l} = q_j\bar{\Phi} _{g\left( j \right)}-p_{g\left( j \right)}-\left[ q_l\bar{\Phi} _{g\left( j \right)}-p_{g\left( l \right)} \right]\\
&=\sum_{k=j}^{l-1}{\left( q_{k+1}-q_k \right) b_{g\left( k+1 \right)}}+\left( q_j-q_l \right) \bar{\Phi} _{g\left( j \right)}\\
&\ge \sum_{k=j}^{l-1}{\left( q_{k+1}-q_k \right) b_{g\left( k+1 \right)}}+\left( q_j-q_l \right) b_{g\left( j+1 \right)}.
\end{aligned}
\end{equation}
If $j = l - 1$, the following inequality holds:
\begin{equation}
\Delta\mathcal{U}_{j,l}\geq \left( q_{j+1}-q_j \right) b_{g( j+1 )} + \left( q_j-q_{j+1} \right) b_{g( j+1 )}\geq 0.
\end{equation}
Otherwise, if $j < l + 1$, via recursive operations, we have
\begin{align}
&\sum\nolimits_{k=j}^{l-1}{\left( q_{k+1}-q_k \right) b_{g\left( k+1 \right)}} \nonumber \\
&\geq \sum_{k=j}^{l-2}{\left( q_{k+1}-q_k \right) b_{g\left( k+1 \right)}} + \left( q_{l}-q_{l-1} \right)b_{g\left( l-1 \right)}\nonumber \\
&\geq \sum_{k=j}^{l-3}{\left( q_{k+1}-q_k \right) b_{g\left( k+1 \right)}} + \left( q_{l}-q_{l-2} \right)b_{g\left( l-2 \right)}\nonumber \\
&\geq \cdots \geq \left( q_l-q_j \right) b_{g\left( j+1 \right)}.
\end{align}
Hence, $\Delta\mathcal{U}_{j,l} \geq 0$ holds. Theorem~\ref{theorem2} is proved.
\end{proof}

\begin{corollary}\label{corollary1}
In the proposed auction mechanism, the truth-telling equilibrium is also an envy-free Nash equilibrium (EFNE).
\end{corollary}
\begin{proof}
As the truth-telling equilibrium satisfies $b_{g(j)} =  \bar{\Phi}_{g(j)}$ and $b_{g(j)} < \bar{\Phi}_{g(j-1)}$, given ${g(j)} = j$ ($1\leq j\leq K$), then we have $\Delta\mathcal{U}_{j,l}\geq 0$ if $j>l$ and $\Delta\mathcal{U}_{j,l}> 0$ if $j<l$. Thereby, the truth-telling equilibrium is an EFNE. Corollary~\ref{corollary1} is proved.
\end{proof}


\begin{corollary}\label{corollary2}
The outcome of the proposed auction mechanism is a stable assignment.
\end{corollary}
\begin{proof}
It suffices to prove that no UAV can gain an improved profit by aborting the assigned UGV in the auction and re-matching with another UGV for battery recharging. According to the Theorem~\ref{theorem2}, in both cases, we can derive $q_j\bar{\Phi} _{g\left( j \right)}-p_{g\left( j \right)}-\left[ q_l\bar{\Phi} _{g\left( j \right)}-p_{g\left( l \right)} \right] \geq 0$ under given constraints. Hence, the assignment outcome produced by our proposed auction mechanism is stable. Corollary~\ref{corollary2} is proved.
\end{proof}

\begin{theorem}\label{theorem5}
The overall computational complexity of the proposed auction mechanism yields $\mathcal{O}\big(I(\tau)\log(I(\tau)) + J(\tau)\log(J(\tau)) + |W(\tau)|^2\big)$.
\end{theorem}
\begin{proof}
In the phase 1 of the auction mechanism, the computational complexity for type evaluation of UAVs and UGVs is $\mathcal{O}(I(\tau) + J(\tau))$. In the next phase 2, the sorting process of UGVs' types and UAVs' valuations yields a complexity of $\mathcal{O}(I(\tau)\log(I(\tau)) + J(\tau)\log(J(\tau)))$, while the allocation process has a complexity of $\mathcal{O}(K)$, where $K = \min\{I(\tau),J(\tau)\}$. In the last phase 3, the payment determination process for winners has a complexity of $\mathcal{O}(|W(\tau)|^2)$, where $|W(\tau)|$ is the number of winners in an auction with time window $\tau$. Thereby, the overall computational complexity yields $\mathcal{O}(I(\tau)\log(I(\tau)) + J(\tau)\log(J(\tau)) + |W(\tau)|^2)$. Theorem~\ref{theorem5} is proved.
\end{proof}

\section{PERFORMANCE EVALUATION}\label{sec:SIMULATION}
In this section, we first introduce the simulation settings, then we discuss the numerical results.
\subsection{Simulation Setup}\label{subsec:evalution1}
We consider a 3D simulation area of $5000\times5000\times10\, \mathrm{m}^3$, where the sensing spot locates at the center of the area and the sensing task area is a circle with radius 200\,m. UAVs are flying over the sensing task area with constant altitudes in $[5,10]$\,m. UGVs are located outside the sensing task area, and their distances to the sensing spot are uniformly distributed in $[0.3,2.5]$\,km. One base station with location $(3000,3000,50)$ offers wireless communication services for UAVs and UGVs at the considered area. The time window in the auction is set as $\tau=8$ seconds. UAV's battery capacity is set as $C_i = 97.58$\,Wh \cite{8758340}, and the alert battery SoC level is $s_{\min}=20\%$. The current battery SoC of UAVs follows the uniform distribution in $[30\%,100\%]$. The wireless energy efficiency of UGVs is set as $\eta_j = 80\%$. The velocity of UGVs varies between $20$\,km/h and $60$\,km/h. The maximum flying velocity of UAV is set as $v_{\max} = 10$\,m/s. UAV's flying power parameters and channel parameters are set according to \cite{7991310,8758340}. The normalized QoRS is adopted, and UGV's QoRS is computed according to its normalized distance to the sensing spot. The rendezvous for each matched pair of UAV and UGV is generated at the midpoint between the sensing spot and the corresponding UGV. The linear function is adopted for UAV's valuation modelling in recharging, i.e., $\Phi_i[t] = \mu_0 + \mu_1 \rho_i[t]$, where the parameters are set as $\mu_0=1$ and $\mu_1=5$. 

We compare the proposed scheme with the following two conventional schemes.
\begin{itemize}
  \item \emph{Exhaustive optimal scheme:} it exhaustively searches the optimal allocation outcomes for UAVs and UGVs in the OCSP problem $\mathbf{P}1$ without the envy-free constraint (\ref{eq:cons7}). 
  \item \emph{Static WPT scheme:} it replaces the UGVs with static WPT facilities in the simulation area for UAV charging services. The auction process between UAVs and static WPT facilities is similar to our proposed auction between UAVs and UGVs.
\end{itemize}

\subsection{Numerical Results}\label{subsec:evalution2}
In Figs.~\ref{fig:simu01}--\ref{fig:simu03}, we evaluate the satisfaction level of UAVs, utility of UAVs, and social surplus in the proposed scheme, compared with the conventional schemes. Then, we evaluate the effect of auction time window in Fig.~\ref{fig:simu03}. After that, in Tables~\ref{simutable1}--\ref{simutable2}, we evaluate the strategy proofness and envy freeness of the proposed scheme. The \emph{satisfaction level of UAVs} is defined as:
\begin{align}
SL \!=\! \sum_{i=1}^{I(\tau)}{\sum_{j=1}^{J(\tau)}{\beta_{i,j}q_j \rho_i(\tau)} },\,\rho_i(\tau) \!=\! \lfloor \frac{\tau}{t} \rfloor^{-1} \cdot \sum_{t\in \tau}{\rho}_i[t].
\end{align}
Besides, the \emph{non-envy ratio} is adopted to measure the envy freeness of the assignment outcomes, which is defined as the number of UAVs that does not envy the other UAV's assignment to the total number of participating UAVs in the auction.

\begin{figure}[!t]
\centering
  \includegraphics[width=8.0cm,height=6.5cm]{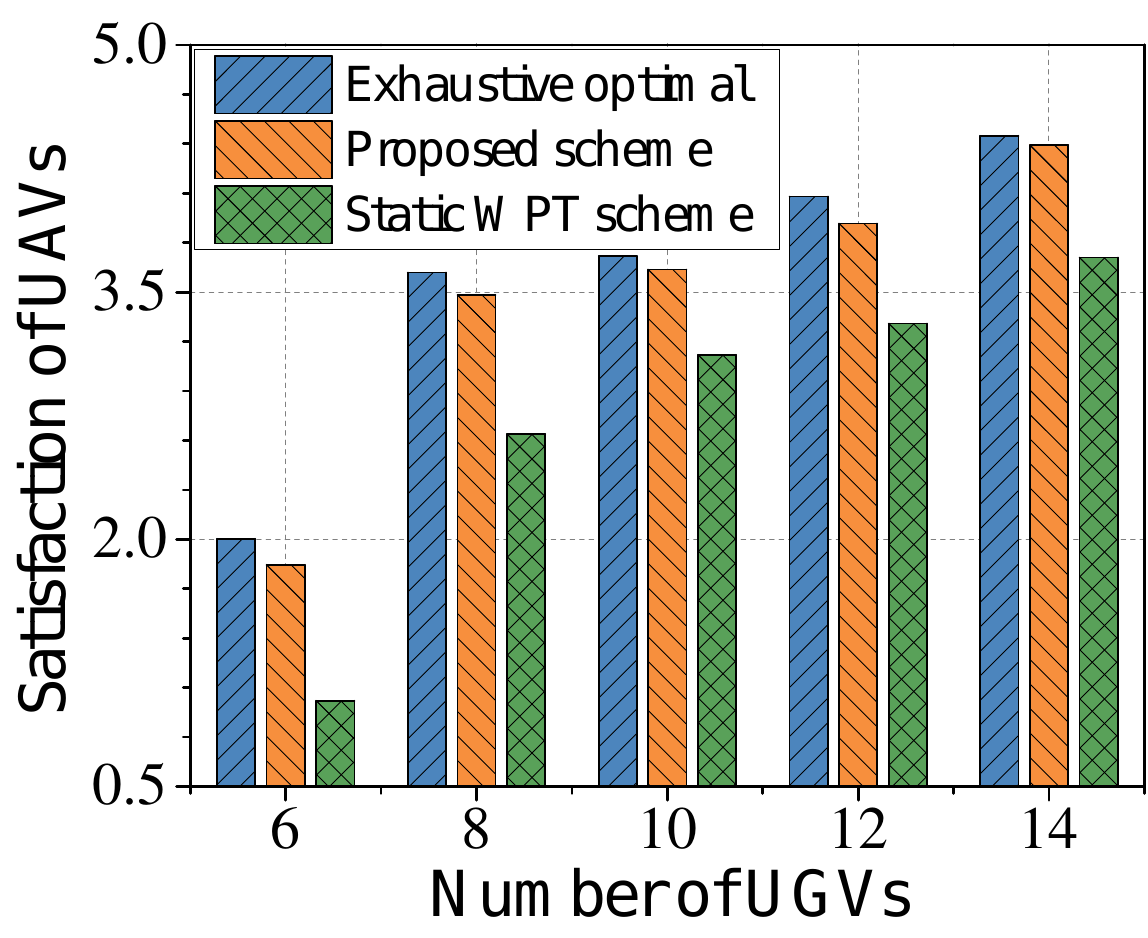}
  \caption{Satisfaction level of UAVs vs. number of UGVs (i.e., $J(\tau)$) in three schemes, where $I(\tau)=10$.} \label{fig:simu01}
\end{figure}

\begin{figure}[!t]
\centering
  \includegraphics[width=8.0cm,height=6.5cm]{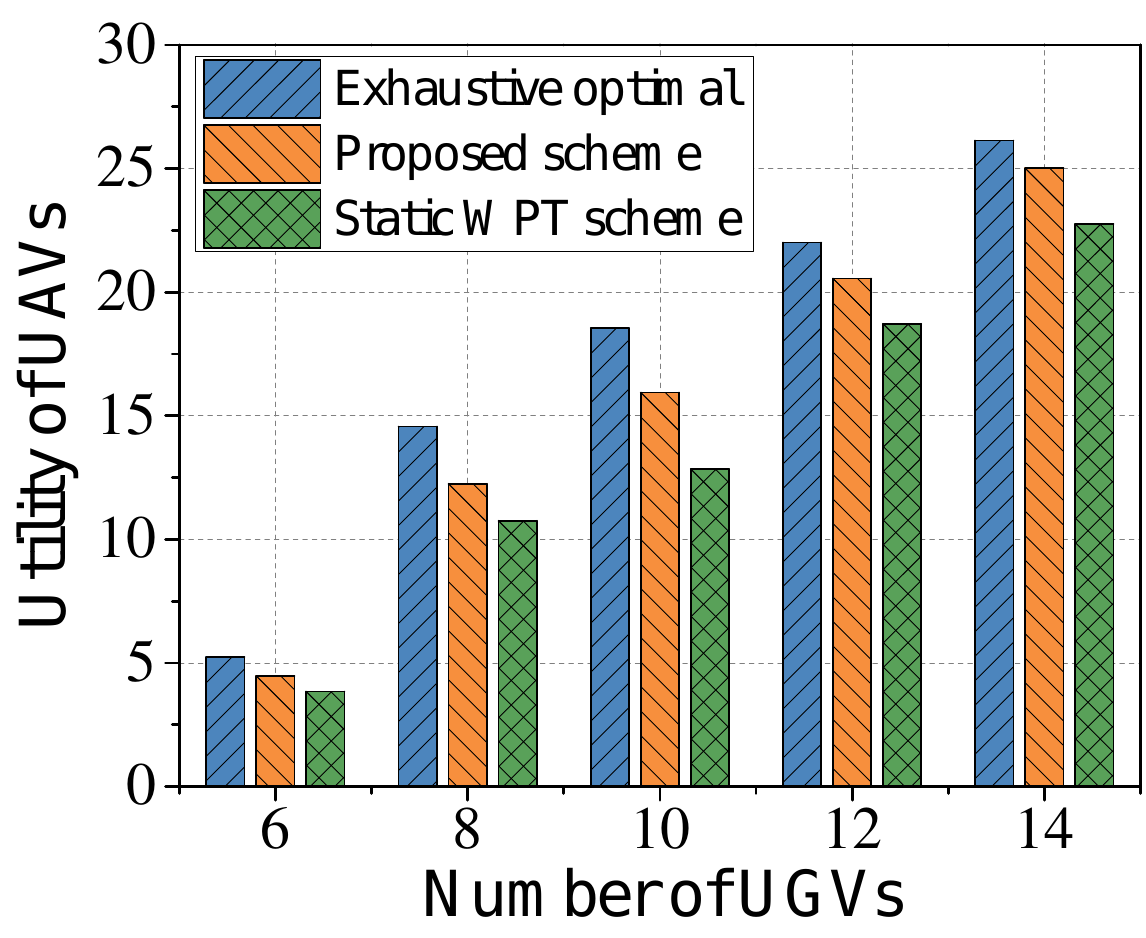}
  \caption{Utility of UAVs vs. number of UGVs (i.e., $J(\tau)$) in three schemes, where $I(\tau)=10$.} \label{fig:simu02}
\end{figure}

Fig.~\ref{fig:simu01} and Fig.~\ref{fig:simu02} show the satisfaction level and utility of UAVs in three schemes, respectively, when the number of UGVs in the auction grows from $6$ to $14$. In these two simulations, the number of UAVs in the auction is set as $10$.
As seen in these two figures, the proposed scheme outperforms the static WPT scheme in attaining a smaller gap with the exhaustive optimal approach. It can be explained as follows.
Compared with the static WPT facilities, our proposed scheme utilizes mobile UGVs to dynamically generate the rendezvouses for UAV launching and battery recharging, thereby saving the flying cost for UAVs. Besides, in both figures, UAVs' satisfaction level and total utility are increasing when the number of UGVs is increasing.
The reason is when more UGVs participating in the auction, each UAV can have a higher chance to match a more preferred UGV for battery recharging. Thereby, UAVs can enjoy a higher satisfaction level of charging service and obtain higher utilities.

\begin{figure}[!t]
\centering
  \includegraphics[width=8.0cm,height=6.5cm]{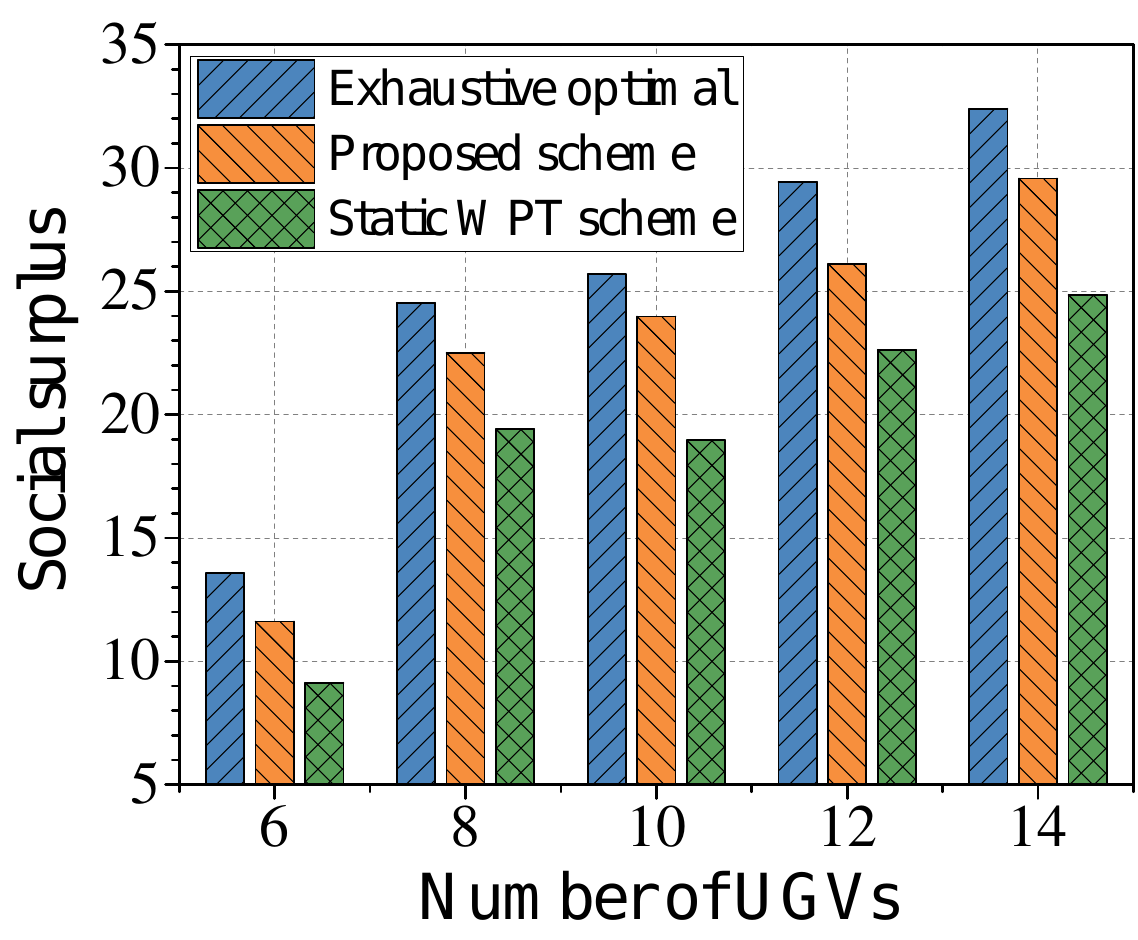}
  \caption{Social surplus vs. number of UGVs (i.e., $J(\tau)$) in three schemes, where $I(\tau)=10$.} \label{fig:simu03}
\end{figure}

\begin{figure}[!t]
\centering
  \includegraphics[width=8.0cm,height=6.5cm]{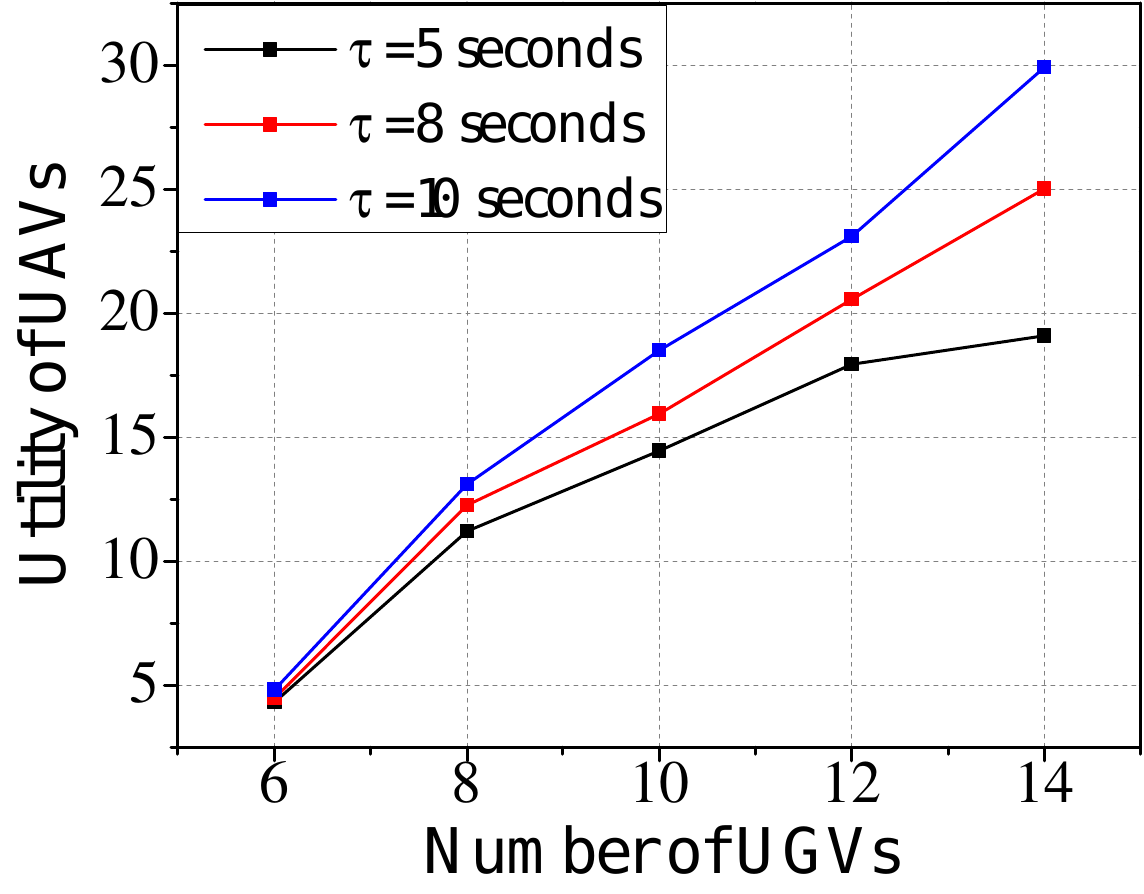}
  \caption{Utility of UAVs vs. auction time window $\tau$ under different numbers of UGVs (i.e., $J(\tau)$), where $I(\tau)=10$.} \label{fig:simu04}
\end{figure}

Fig.~\ref{fig:simu03} depicts the social surplus in three schemes, where the number of UGVs in the auction varies between $6$ and $14$. Here, the number of UAVs in the auction is $10$. From Fig.~\ref{fig:simu03}, it can be observed that the proposed scheme attains a higher social surplus than the static WPT scheme. The reason is that in the static WPT scheme, the WPT facilities in the simulation area are static, causing a higher round-trip cost for UAVs than our UGV-assisted WPT scheme. Additionally, given more UGVs, as the chances for both UAVs and UGVs to match their more preferred parter can be higher, the overall utility of UAVs and UGVs (i.e., social surplus defined in Eq. (\ref{eq4-3})) is greater.

Fig.~\ref{fig:simu04} illustrates the utility of UAVs in the proposed scheme when both the auction time window $\tau$ and number of UGVs vary. As seen in Fig.~\ref{fig:simu04}, the longer auction time window can result in higher UAV utilities when the number of UGVs is fixed. The reason is that, in our proposed auction scheme, the bid collection process ends if the time window $\tau$ elapses. As such, given the longer auction time window, more bids of UAVs and UGVs can be included in the current auction to help them make better energy matching choices. Besides, the utility of UAVs grows with the number of UGVs, which has been analyzed in Fig.~\ref{fig:simu02}.

\begin{table}[!t]\setlength{\abovecaptionskip}{0.1cm}
    \centering
    \caption{Comparison of the utility of a randomly selected UAV under truthful bidding and untruthful bidding in different auction sizes.}\label{simutable1}
    \begin{tabular}{|c|c|c|}
        \hline
        {Auction size} & \textbf{\begin{tabular}[c]{@{}c@{}}Truthful\\bidding\end{tabular}} & \textbf{\begin{tabular}[c]{@{}c@{}}Untruthful\\bidding\end{tabular}} \\ \hline
        ($I(\tau)=5,J(\tau)=5$)      & \textbf{3.874}          & 2.0148            \\ \hline
        ($I(\tau)=20,J(\tau)=20$)      & \textbf{28.234}         & 22.041            \\ \hline
    \end{tabular}
\end{table}

\begin{table}[!t]\setlength{\abovecaptionskip}{0.1cm}
    \centering
    \caption{Comparison of the ratio of non-envy UAVs in three schemes under small-scale and large-scale auctions.}\label{simutable2}
    \begin{tabular}{|c|c|c|}
        \hline
       Auction size                                          & \textbf{\begin{tabular}[c]{@{}c@{}}Exhaustive \\ optimal\end{tabular}} & \multicolumn{1}{l|}{\textbf{Ours}} \\ \hline
        ($I(\tau)=5,J(\tau)=5$)   & 40\%                                             & \textbf{100\%}                                         \\ \hline
        ($I(\tau)=20,J(\tau)=20$) & 30\%                                             & \textbf{100\%}                                         \\ \hline
    \end{tabular}
\end{table}

Table~\ref{simutable1} compares the UAV utility under truthful bidding and untruthful bidding in the proposed auction. As seen in Table~\ref{simutable1}, the utility of the randomly selected UAV when bidding truthfully is greater than that in untruthfully bidding under both small-scale and large-scale auctions. It indicates that bidding truthfully is the dominant strategy for UAVs, which validates the strategy proofness of our auction mechanism and conforms to Theorem~\ref{theorem1}.

Table~\ref{simutable2} compares the ratio of non-envy UAVs in two schemes under small-scale auction (i.e., $I(\tau)=5,J(\tau)=5$) and large-scale auction (i.e., $I(\tau)=20,J(\tau)=20$). As observed in Table~\ref{simutable2}, the proposed scheme outperforms the exhaustive optimal approach and enforces envy freeness for all UAVs under both small-scale and large-scale auctions, which also conforms to the theoretical results in Theorem~\ref{theorem2}.

\section{CONCLUSION}\label{sec:CONSLUSION}
UAV's limited flight endurance is one of the main impediments to modern UAV applications. By leveraging ground vehicles mounted with WPT facilities on their proofs, this paper has proposed a mobile and collaborative recharging scheme for UAVs to facilitate on-demand wireless battery recharging. An energy scheduling problem for multiple UAVs and multiple vehicles has been formulated under practical constraints and energy competitions in the highly dynamic network. We have also devised an online auction-based solution with low complexity to allocate and price idle wireless chargers on vehicular proofs in real time. Theoretical analyses have proved that the proposed scheme produces strategy-proof, envy-free, and stable allocation outcomes. Lastly, numerical results validate the effectiveness of the proposed scheme in delivering more satisfactory UAV charging services. For the future work, we plan to improve the auction efficiency and investigate the charging scheduling mechanism under more complex situations.

\bibliographystyle{IEEETran}
\bibliography{myref}

\end{document}